\documentclass{lmcs}

\usepackage{hyperref}

\def\eg{{\em e.g.}}

\def\ie{{\em i.e.}}

\usepackage{amssymb}
\usepackage{csquotes}
\usepackage{listings}

\newcommand*{\casesif}{&\text{if }}
\newcommand*{\casesowise}{&\text{otherwise}}
\newcommand{\Set}[1]{\lbrace\,#1\,\rbrace}
\newcommand*{\function}[3]{#1 \colon #2 \rightarrow #3}
\newcommand*{\catcompose}{\mathbin{;}} 
\newcommand*{\andalso}{\wedge}
\newcommand*{\orelse}{\vee}
\DeclareMathOperator{\suchthat}{\mid}
\newcommand{\tuple}[1]{(#1)}
\newcommand*{\Naturals}{\omega}

\newcommand*{\interpretation}[1]{\lVert #1\rVert}

\newcommand*{\true}{\mathrm{t}}
\newcommand*{\false}{\mathrm{f}}
\newcommand*{\meaningless}{\mathord{?}}
\newcommand*{\emptystring}{\mathtt{e}}

\newcommand*{\bottom}{\mathord{\perp}}

\newcommand*{\truth}[3]{u^{#1, #2}_{#3}}
\newcommand*{\program}[3]{p^{#1, #2}_{#3}}
\newcommand*{\programs}[3]{P^{#1, #2}_{#3}}
\newcommand*{\validPrograms}[3]{V^{#1, #2}_{#3}}
\newcommand*{\nonProgram}[3]{\bottom^{#1, #2}_{#3}}
\newcommand*{\interpretationSet}[3]{I^{#1, #2}_{#3}}
\newcommand*{\simulation}[3]{\simulatedBy^{#1, #2, #3}}
\newcommand*{\semantics}[4]{\interpretation{#1}^{#2, #3}_{#4}}
\newcommand*{\hypothesis}[3]{h^{#1, #2}_{#3}}
\newcommand*{\action}[3]{a^{#1, #2}_{#3}}
\newcommand*{\evidence}[3]{e^{#1, #2}_{#3}}
\newcommand*{\lang}[3]{L^{#1, #2}_{#3}}
\newcommand*{\theory}[4]{T^{#1, #2, #3}_{#4}}

\newcommand*{\representation}[6]{f^{#1\rightarrow #2, #3\rightarrow#4}_{#5\rightarrow#6}}

\newcommand*{\implication}[5]{#3, #1 \models^{#4}_{#5} #2}

\DeclareMathOperator*{\limBold}{\mathbf{lim}}
\newcommand*{\trueInLimit}[2]{\mathbf{\limBold}_{#2\rightarrow \infty} #1}

\DeclareMathOperator*{\correctOp}{Correct}
\newcommand*{\isCorrect}[3]{\correctOp^{#1, #2}_{#3}}
\DeclareMathOperator*{\validOp}{Valid}
\newcommand*{\isValid}[3]{\validOp^{#1, #2}_{#3}}
\DeclareMathOperator*{\completeOp}{Complete}
\newcommand*{\isComplete}[3]{\completeOp^{#1, #2}_{#3}}
\DeclareMathOperator*{\matureOp}{Mature}
\newcommand*{\isMature}[3]{\matureOp^{#1, #2}_{#3}}
\newcommand*{\achievesMaturity}[3]{\trueInLimit{\isMature{#1}{#2}{#3}}{#3}}

\newcommand*{\representable}[6]{\left\langle #1\right\rangle_{#3}^{#5} \Rightarrow_\mathrm{repr} \left\langle #2 \right\rangle_{#4}^{#6} }

\newcommand*{\simulatedBy}{\preccurlyeq_{\text{sim}}}
\newcommand*{\bisimilar}{\approx_{\text{sim}}}

\keywords{logic of software development, programs, specifications, translatability, methodology of software development, possible worlds, program generators}

\begin{document}
\title[Software development by multiple agents following a common specification]{A logical framework to model software development by multiple agents following a common specification}

\thanks{The first author has been funded by the Special Account for Research Funding (E.L.K.E.) of the National Technical University of Athens. The second author has been supported by the PEVE2021 action of the Special Account for Research Funding (E.L.K.E.) of the National Technical University of Athens. This work is dedicated to the late Aristidis Arageorgis}

\author[G. V. Pitsiladis]{Georgios V. Pitsiladis\lmcsorcid{0000-0003-2279-7738}}
\author[P. S. Stefaneas]{Petros S. Stefaneas\lmcsorcid{0000-0002-2096-9914}}
\address{School of Applied Mathematical and Physical Science, National Technical University of Athens, Greece}
\email{gpitsiladis@mail.ntua.gr, petros@math.ntua.gr} 
 
\begin{abstract}
	\noindent In this paper, we address program development by multiple different programmers (or programming teams), each working in different settings (programming languages or reasoning frameworks), but following a common specification; in particular, we examine at an abstract level the problem of translatability between their produced programs.
	To this end, after consideration of some philosophical issues regarding program development, including its similarities and dissimilarities with  scientific theorising, we extend a logical framework built to describe scientific theorising in relativist settings: our extensions add the ability of reasoning about programs, the iterative process of their generation, and their specifications.
	We are thus able to define a notion of translation between the outputs of program generators and prove that there is a (trivial) such translation when two program generators follow the same specification reliably (in a specific sense of reliability).
\end{abstract}

\maketitle

\section{Introduction}

In~\cite{arageorgis_relativism_2017}, A. Arageorgis, building on the work of~\cite{kelly_reliability_1992}, constructs a formal framework describing scientific theorising in the presence of many variants of relativism: truth, meaning, logical framework, and evidence are expected to depend on the actions, conjectures, and conceptual choices of each scientist (or research programme) at any given time. Arageorgis proves that, if two scientists (or research programmes) start with common background knowledge and work in a way that asymptotically brings each of them close to their respective truth about the world, then there will asymptotically be a (trivial) translation between their resulting theories.

It has been argued~\cite{plato_philosophy_2021} that software development has an empirical character, albeit with some important differences compared to the empirical character of natural science. Moreover, computer systems can be considered as technical artefacts~\cite{turner_computational_2018}; in that sense, they require the use of some background scientific knowledge in order to be built~\cite{arageorgis_demarcating_1989}. Both these considerations lead us to believe that the framework defined in~\cite{arageorgis_relativism_2017} can be (carefully) modified to reason fruitfully about program development; in particular, given its translatability result and, more generally, its strength in modelling situations where multiple agents work independently but share some background assumptions, we feel that it can shed some light on similar concerns in computer science.

With the above in mind, we extend the framework of~\cite{arageorgis_relativism_2017} to describe software development. We reinterpret scientists (or research programmes) as software developers (or teams). The notions of relativised truth, meaning, logical framework, and evidence remain present in our generalisation.  We add that each developer (or team) produces a piece of code (a technical artefact) at every time instant; hence, the relativised setting is extended with program semantics, allowing different programmers to work in different programming languages, or even programming paradigms, resulting in different program semantics.

It is then possible to define a notion of translatability between the outputs of program generators. 
As in~\cite{arageorgis_relativism_2017}, we manage to obtain a proof that if two programmers (or teams) start with common specification and work in a way that asymptotically makes each one write correct programs relative to their shared specification, then, there will asymptotically be a (trivial) translation between their resulting programs and theories.

The rest of the paper is organised as follows: Section~\ref{section:philosophy-of-cs} contains philosophical discussion on the ontology of programs and specifications and on the methodology of software development; the purpose of this discussion is threefold: (a) to argue on the feasibility of porting results form the philosophy of natural science to reason about computer science, (b) to support the conceptual choices made while adapting~\cite{arageorgis_relativism_2017} and developing our framework, and (c) to provide an overview of the aspects of the development of computational systems, especially in relation to programs and specifications, from which the aspects that can or cannot be handled by our framework can be inferred.
Section~\ref{section:logic-of-program-development} contains a detailed presentation of our framework.
Finally, Section~\ref{section:conclusion} contains some concluding remarks.

\section{Philosophical considerations on program development} \label{section:philosophy-of-cs}

The development of computational systems has been characterised as an empirical discipline~\cite[Section~8.3]{plato_philosophy_2021} and parallels can be drawn between (parts of) its methodology and elements of the scientific method~\cite{angius_problem_2014}. This is due to the need for verifying whether programs adhere to their specification \textemdash a need that quickly becomes all the more urgent as programs become more complex\textemdash, which is pursued by forming (maybe formal) theories of the behaviour of the program and by observing executions of the program with carefully selected input; both these aspects bear resemblances with scientific theorising and experimentation.

Moreover, as noted by~\cite{angius_problem_2014}, creating complete computational systems definitely requires solving several problems pertaining to an array of scientific disciplines: it is a \emph{scientifically attested technology} in the sense of~\cite{arageorgis_demarcating_1989}.

However, arguably, the analogy between software engineering and experimental science has a number of peculiarities; \cite{plato_philosophy_2021} mentions some of them, the most important being the difference between specifications and scientific theories that we discuss in Section~\ref{section:philosophy-specifications} below.

Thus, adapting the framework of~\cite{arageorgis_relativism_2017}, designed for the description of scientific theorising in the \enquote{hard sciences}, to obtain a framework that describes program development might be sensible, but, in any case, it must be carried out with care, in order to accommodate the differences between the two kinds of disciplines.

Our treatment of software development will be based on a computer science perspective, which, as argued by~\cite{gruner_problems_2011}, is different from (although related to) a software engineering perspective. For example, we will ignore aspects of software such as its cost, legal considerations (at least those not directly expressed in the specification), and project management, \ie,  those aspects  termed the \emph{socio-technical modality} by~\cite{de_mol_what_2018}. We will also ignore aesthetic qualities of programs, such as form and legibility; moreover, we will not be concerned with simplicity (although, according to~\cite[Section~2.3]{turner_computational_2018}, it contributes to correctness). 

\subsection{Levels of abstraction}
In \cite{primiero_information_2016}, building on the method of abstraction~\cite{floridi_method_2016}, the basic levels of abstraction (LoAs) in work when creating computational systems are identified and discussed. They range from structured physical data (for example, bits stored in a hard drive) to intentional information content (this corresponds to \emph{intention} of~\cite{turner_computational_2020} or, maybe more loosely, to \emph{requirements} of~\cite{zave_four_1997}; for example, what kind of actions a travel agent expects a booking platform to perform).

A full account of development of computational systems as complete entities should take into account all the LoAs that are employed when building such systems.
In this work, we will not explore the full LoA hierarchy; for example, we will not consider the actual circuits running a program (also known as the \emph{physical modality} of~\cite{de_mol_what_2018}). We will focus on programs and on their specifications; in this context, a program is instructional information in some programming language, while a specification is abstract (but still formal) information which is meant to capture the intention of the computational system stakeholders.

As stressed by~\cite{floridi_method_2016}, LoAs require an analysis exhibiting \emph{relationism}: pluralistic but without endorsing relativism. One of the goals of our work is to express situations where two programmers (or programming teams) share a part of the LoA hierarchy (in particular, they share a common specification) but might be working with different elements in the lower LoAs, hence we introduce a kind of \enquote{horizontal} relationism among the different agents, in addition to the \enquote{vertical} relationism among the LoAs mentioned above. The framework of~\cite{arageorgis_relativism_2017} seems fitting for that goal, since it was specifically designed for a pluralistic analysis.

\subsection{Programs: syntax and semantics}

Programs are (sequences of) instructions given in some programming language. As such, they are expressed as strings in some well-defined alphabet and they must conform to the (formal) syntactic rules of their language.

In~\cite{primiero_information_2016}, it is stressed that programs have no alethic values: their meaning ultimately lies on their ability to control the operation of the physical components of a machine, although this is mediated by the (abstract) semantics of the programming language. Similarly, there are two main ways to study the semantics of programming languages and programs: operational semantics, which focuses on how each step of the program operates, and denotational semantics, which focuses on the program as a whole; for a comparison and other notions of semantics, see~\cite{turner_computational_2018}.

While not having alethic values themselves, programs, being technical artefacts, have both functional and structural properties~\cite[Section~3]{angius_logic_2018}: in short, functional properties describe \emph{what functionalities} a program implements, while structural properties describe \emph{how} it operates. Reasoning on such properties is therefore possible, for descriptive or for normative purposes.
When programs are too complex, reasoning on their functional properties might be accomplished by constructing models to act as their proxies, in order to abstract away some of the complexity \cite[Section \emph{Model Construction}]{turner_specification_2011}.

\subsection{Specifications and verification}\label{section:philosophy-specifications}

Specifications\footnote{As discussed in~\cite[Section~6.1]{angius_logic_2018}, there are two notions of \enquote{specifications} in the literature: \emph{property} specifications (called \emph{functional} specifications in~\cite{turner_specification_2011}) and \emph{system} specifications. The former describe what functional properties a program \emph{should} have, while the latter describe what properties an \emph{already constructed} program has. The notion of \enquote{specifications} in this paper refers to the first kind.} define how a program should function: they describe the functional properties that produced programs should possess. Their articulation ranges from natural language to specialised formal languages. Similarly (but not necessarily in parallel), their level of detail and their level of formality can vary.

What is common to all specifications is their relation to programs: they form a \emph{dominant/submissive} pair, in the sense of~\cite{turner_computational_2020}. In the words of~\cite{turner_specification_2011}, \enquote{it is the act of taking a definition to have normative force over the construction of an artefact that turns a mere definition into a specification}, or, more plainly, \enquote{something is a specification when it is given correctness jurisdiction over an artefact}. This characteristic does not stem from the form or the extensional content of the specifications, but from the intention of their use as such~\cite{turner_computational_2020}.

This is a key difference between specifications and scientific theories. While it is in general true that specifications describe programs and scientific theories describe the world, they \enquote{describe} in different senses. The latter have explanatory and predictive character, while the former are normative (any predictive or explanatory characteristics a specification might have are incidental~\cite[Section~\emph{Specification and Theory}]{turner_specification_2011}). Accordingly, when there is a mismatch between a scientific theory and the world, then the theory must be revised; on the contrary, when there is a mismatch between a specification and a program, then the program must be revised. \label{text:normative-force-opposed-to-science}

Notice that, as stressed by~\cite{turner_specification_2011}, the relationship between programs and their specifications is more complicated in practice. On the one hand, specifications are not always correct themselves or may even be unrealistic; these and other reasons imply that there are cases where a program/specification misalignment causes the specification to be revised instead of the program. On the other hand, especially in big systems, where different parts of the specification have different degrees of importance, a program might be accepted if it is \enquote{fit for use}, in the sense that it matches the specification enough to be used for its (main) intended purpose, but not (necessarily) completely. We will not tackle these considerations here, but we will hint on how we believe they could be tackled.

Since specifications are truth-determining for programs (and, dually, programs are correct in view of specifications), there arises the problem of checking whether the created program has the intended properties: the problem of verification. Verification of some sort, formal (theorem proving, model checking), experimental (testing), or mixed is crucial in the development process. In fact, \cite{angius_mutual_2020} argues that formal and experimental verification are intertwined.

Often, when verifying that a program meets its specification, possible failures in lower LoAs are tacitly or explicitly ignored \cite[Section~\emph{Fairness Constraints as Ceteris Paribus Modifiers}]{angius_abstraction_2013}\cite[Section~\emph{Verifying Programs}]{turner_computational_2020}. Since we are focusing on the LoAs of programs and specifications, our framework can be considered to incorporate such simplifications in its design.

\section{A logic of software development}\label{section:logic-of-program-development}

Fix a finite alphabet and let $S$ be the set of all strings that may be written with characters of the alphabet; of course, $S$ is countably infinite. Let $\emptystring \in S$ be the empty string.

\begin{nota}
If $\function{f}{A}{B}$ and $\function{g}{B}{C}$ are functions, then $f\catcompose g$ is their composition.
\end{nota}

\begin{nota}
$\Naturals$ is the set of natural numbers with their usual order. For a set $X$, $X^{\Naturals^\star}$ is the set of all finite sequences of $X$; we accept that $\emptyset\in X^{\Naturals^\star}$ by convention.
\end{nota}

\begin{nota}
Let $\mathbf{P_n}$ be a predicate that depends on a parameter $n\in\Naturals$. We can then form the predicate \enquote{$\mathbf{P_n}$ in the limit},
\[
\trueInLimit{\mathbf{P_n}}{n} \equiv \text{there exists some } n \in \Naturals \text{ such that } m\geq n \text{ implies } \mathbf{P_m}.
\]
\end{nota}

\begin{rem}
Let $\mathbf{P_n}$, $\mathbf{Q_n}$ be predicates. Then, $\trueInLimit{\left(\mathbf{P_n} \andalso \mathbf{Q_n}\right)}{n}$ iff $\trueInLimit{\mathbf{P_n}}{n}$ and $\trueInLimit{\mathbf{Q_n}}{n}$.
\end{rem}

In the original framework of~\cite{arageorgis_relativism_2017}, the objective part of truth and meaning is modelled via possible worlds, a collection of which constitutes background knowledge; at each stage of scientific inquiry, hypotheses generators (\ie, scientists or research programmes), after examining previous evidence given to them by the world, output a hypothesis (\ie, a conjecture based on their conceptual choices) and an action, to which each possible world responds with fresh evidence and with a valuation of what they can consider true, false, or meaningless based on their selected logical framework. Notions of language and theory are defined on top of this interaction for every stage of inquiry. Using these elements, some properties are stated that aim to express the correctness and completeness of theories. Also, a notion of translatability between theories is defined.

The notions of hypotheses, and actions, valuations, and evidence \textemdash and hence also reasoning languages and theories\textemdash remain present in our generalisation, since program generators (\ie, programmers, programming teams, or even automatic code generators) entertain some hypotheses regarding the behaviour of their produced code and its (natural as well as artificial) environment.  We add some code is produced at every time instant; hence, the response of the possible worlds needs to be extended with program semantics, allowing different programmers to work in different programming languages, or even programming paradigms. The notion of background knowledge is reinterpreted to model specifications regarding the deliverable piece of code.
Moreover, we state some properties that aim to express the well-formedness of a program and its correctness relative to its specification. We also extend the notion of translatability so as to include programs.

\subsection{Programs}
Let $P$ denote all possible programs; notice that $P$ may be allowed to not contain only valid programs, but also any code that might be produced by programmers (for example, code with typos). Since programs are actually strings, it is enough to define $P\subset S$ for this presentation; however, notice that programs could have more structure than just plain strings ($P$ could, for example, include structured metadata along with the code of each program). We will need $P$ to contain an empty program; we will use the empty string for that purpose, \ie, $\emptystring\in P$.

Regarding the semantics of programs, we do not want to restrain ourselves to some particular variant (\eg, operational, denotational) or programming paradigm (\eg, procedural, functional), hence we will use a general notion of program semantics as a function that maps programs to arbitrary mathematical objects. As argued in~\cite[Chapter~10]{turner_computational_2018}, this formalism, although usually stated in the context of denotational semantics, can also incorporate operational semantics (by mapping programs to functions on states).

We are also interested in comparing the high-level behaviour of programs, in order to detect whether a program $P^\prime$ manifests (at least) the behaviours of another program $P$. This can be achieved by a relation of \emph{simulation}, or \emph{inexact copy} in the sense of~\cite{angius_logic_2018}, which uses a notion of state-based simulation that results in a reflexive and transitive relation, \ie, a preorder; here, we will only assume the algebraic properties of simulation without being confined to a state-based definition.

\begin{defi}
A \emph{semantical framework} is a tuple $\tuple{I, \bottom, \interpretation{\cdot}, \simulatedBy}$ where $I$ is an arbitrary set of semantical objects, $\function{\interpretation{\cdot}}{P}{I}$ maps each program to its \emph{semantics}, $\bottom \in I$ is the value to which non-valid programs are mapped, and the \emph{simulation relation} $\simulatedBy\subseteq I\times I$ is a preorder on $I$ with $\bottom$ as its minimum element. The empty program shall be treated specially and be mapped to $\bottom$, \ie, $\interpretation{\emptystring}=\bottom$. Let $\mathcal{P}$ be the set of all semantical frameworks.
\end{defi}
\begin{nota}
Let $\tuple{I, \interpretation{\cdot}, \simulatedBy}$ be a semantical framework. We extend the notation $\simulatedBy$ to programs by defining, for $p_1, p_2\in P$, that $p_1\simulatedBy p_2$ iff $\interpretation{p_1}\simulatedBy \interpretation{p_2}$. Also, we write $x \bisimilar y$ iff $x \simulatedBy y$ and $y\simulatedBy x$ (alluding to the notion of bisimulation).
\end{nota}

\subsection{Reasoning about programs and the world}

In our setting, verification of programs will be modelled via an (abstract) notion of reasoning on program properties. This is more closely related to formal verification, but, due to the abstractness of our definitions, it can also be made to include descriptions of experimental verification results.

Let $H\subset S$ be a designated subset of $S$ that contains hypotheses. We can think of these hypotheses as descriptions of program properties and functionalities; as such, they can be used for modelling notions such as functional specifications or formal reasoning on programs and properties of the program's environment. In an extended framework which would include physical machines, hypotheses might be extended to also include relevant statements about the physical world.

A \emph{truth assignment} is any function $\function{u}{H}{\Set{\true, \false, \meaningless}}$ ($\true$ for true, $\false$ for false, and $\meaningless$ for meaningless) such that $u(\emptystring) = \meaningless$. Let $\mathcal{U}$ be the set of all truth assignments.

Although we are using only three truth values, it could be conceivable to extend the framework to use more truth values (perhaps organised in a bilattice~\cite{fitting_bilattices_2006}), in an attempt to model situations where programs might be accepted as \enquote{fit for use}.

\subsection{Specifications}

Let $A$ be an arbitrary set of possible actions and $E$ be an arbitrary set of possible evidence. Actions encode voluntary acts of software developing that are under the control of the programmers. Evidence encode what can be observed by the programmers. These sets will be different in any application of our framework: for example, actions might contain the data of verification attempts, encodings of some actual actions of the programmer, etc; evidence might contain verification results (as an alternative of encoding them into the hypotheses), encodings of (past or current) events, etc.

A central notion in our definition of specification will be that of a possible world. In the words of~\cite[Section~III]{kelly_reliability_1992}, a possible world \enquote{is the component of truth that cannot be manipulated at will}. Hence, while the programmer can (in principle) code, assume, and do whatever she pleases, the world poses some semantical (and practical) limitations. This reflects the discussion in Section~\ref{text:normative-force-opposed-to-science} regarding the normative force of specifications over programs, especially when taking into account that in the framework of~\cite{arageorgis_relativism_2017} possible worlds have a similar force against scientific theories.

The term \enquote{possible world} is kept unchanged with respect to~\cite{arageorgis_relativism_2017} for one more reason: as argued in~\cite[Section~3.2]{zave_four_1997}, specifications can be thought of as containing only statements about the environment of the machine, either per se or in relation to the machine; such a viewpoint is meant to aid in avoiding implementation bias during the formulation of specifications. The term \enquote{possible world} can serve as a subtle pointer to that viewpoint.

\begin{defi}
	A \emph{possible world} is a function $\function{w}{P\times H\times A\times\Naturals}{\mathcal{P}\times\mathcal{U}\times E}$.
\end{defi}

Hence, at each stage of software development, the world responds to the current program, hypothesis, and action of the programmer by imposing the semantics of the program, specifying what can be considered true, false, or meaningles (for this specific programmer), and outputting a piece of evidence. Notice that this is intended as a behavioural, not ontological, description of our possible worlds; in other words, it is an \emph{interface} they must adhere to, but the inner structure of a possible world might be different than a function.

The fact that truth assignments and program semantics are decided by the possible world at each time instance \emph{depending on the actions of the program generators} gives the ability to incorporate what we dubbed \emph{horizontal relationism}, because each program generator might perform actions that end up in different truth or semantics; however, since they are controlled by the world and not completely by the program generator, truth, meaning, and evidence also have an objective side, reflecting the fact that these elements depend also on factors not under the control of the software developers.

This definition does not use the most generic signature of possible worlds of~\cite[Section~III.A]{kelly_reliability_1992}; such a signature would have as part of its input an infinite sequence of $P\times H\times A$ tuples instead of a single $P\times H\times A$ tuple, so as to allow the output to depend on the past actions of the programmer (the part of the sequence with indexes before the time argument) and the future actions of the programmer (the part of the sequence with indexes after the time argument).
However, for our purposes, the present status seems to be enough, since possible worlds will be used as building blocks of specifications. We consider specifications to be independent of previous (and definitely of future) actions of programmers. Moreover, when programs are compared to their specifications, they are compared to the \emph{current} state of the specification. Note that, if we wished to model specification change (and to do so accurately), we might have needed to resort to dependence on the past, hence replace the $P\times H\times A$ tuple with a finite sequence (having length equal to the time argument) of $P\times H\times A$ tuples.

In the original framework of~\cite{arageorgis_relativism_2017}, a set of possible worlds defines the \emph{background knowledge} shared by two different generators of scientific theories; the intuition behind this definition is that a property about the actual world is equivalent to the set of possible worlds that exhibit it. We will reinterpret the same construct, sets of possible worlds, as specifications; similarly, the intuition will be that a property about the program to be created is equivalent to the set of the worlds that specify it.

\begin{defi}
	A \emph{specification} $K$ is a set of possible worlds.
\end{defi}

Each possible world in a specification relaxes the behaviour of the program in some (possibly different) ways; for example, world $w_1$ might allow the program to be written in more programming languages, $w_2$ might omit a temporal liveness property such as \enquote{for each user, a logout happens at most 5 minutes after a login}, $w_3$ might omit a temporal safety property such as \enquote{no user that does not have the role \lstinline|Admin| can access the control panel}, etc (for a definition of safety and liveness properties, see, for example,~\cite{angius_copying_2022}). Hence, the intersection of all the properties that are present in all the worlds corresponds to our usual intuition about specifications; this statement will have its mathematical counterpart in the definition of theories below.

The fact that the output of each possible world depends on time allows for a (basic) form of specification change. Notice that this change does not depend on the actions of any specific programmer; it is defined independently of the interaction between world and programmer. Hence, specifications change uniformly for all programmers following them.

\subsection{Developing programs}

We now turn to describing programmers. We will use the term \emph{program generator}, in order to also be able to accommodate other kinds of entities, such as teams of programmers or even automatic code generators. As with possible worlds, our definition is to be read behaviourally, not ontologically.

\begin{defi}
A \emph{program generator} is a function $\function{\delta}{E^{\Naturals^\star}}{P\times H\times A}$.
\end{defi}

The structures of possible worlds and program generators allows them to interact.
At each point $n$ in time, a program generator $\delta$ is free to review the historical evidence and, having done so, produces a program, accompanied by a hypothesis regarding its functionality and its properties, and performs some action. The world $w$ can then respond and, in the process, create some fresh evidence. This allows the program generator to produce new output at the next point in time. This process can be summarised in mathematical notation as a sequence 
$\tuple{
\program{w}{\delta}{n},
\hypothesis{w}{\delta}{n},
\truth{w}{\delta}{n},
\interpretationSet{w}{\delta}{n},
\nonProgram{w}{\delta}{n},
\semantics{\cdot}{w}{\delta}{n},
\simulation{w}{\delta}{n},
\action{w}{\delta}{n},
\evidence{w}{\delta}{n}
}_{n\in\Naturals}$, where
\begin{gather*}
\left(\program{w}{\delta}{0}, \hypothesis{w}{\delta}{0}, \action{w}{\delta}{0}\right) = \delta\left(\emptyset\right),\\
\left( \tuple{\interpretationSet{w}{\delta}{n}, \nonProgram{w}{\delta}{n}, \semantics{\cdot}{w}{\delta}{n}, \simulation{w}{\delta}{n}},\truth{w}{\delta}{n}, \evidence{w}{\delta}{n}\right) = w\left(\program{w}{\delta}{n}, \hypothesis{w}{\delta}{n}, \action{w}{\delta}{n}, n\right),\\
\left(\program{w}{\delta}{n+1}, \hypothesis{w}{\delta}{n+1}, \action{w}{\delta}{n+1}\right) = \delta\left(\left(\evidence{w}{\delta}{m}\right)_{m\leq n}\right).
\end{gather*}

\begin{exa}
	Consider a setting where a programmer $\delta$ is using git~\cite{git} to store code and wants to write a web application that has a control panel. According to the specification, the web application must have a \lstinline|/login| endpoint where users can login; each user has a role and users that do not have role \lstinline|Admin| cannot access the control panel. Hence, each world of the specification contains these properties, however different worlds might also allow more behaviours.
	
	$\delta$ has decided to work with the PHP programming language and validate their program using linear temporal logic (which happens to be supported by the specification) and automated tests. For each step, the current program is a string of (hopefully well-formatted) PHP code; the hypothesis contains a conjunction of the following properties (possibly among many other properties), formally expressed in the selected logic:
	\begin{itemize}
		\item For each request to endpoint \lstinline|/login|, the response status is not an \lstinline|HTTP 404| \lstinline|Not found|.
		\item For each request to endpoints that start with \lstinline|/admin/|, if the requesting user does not have role \lstinline|Admin|, then the response status is an \lstinline|HTTP 403| \lstinline|Forbidden| or an \lstinline|HTTP 401| \lstinline|Unauthorized|.
	\end{itemize}
	Moreover, for each step, the action of $\delta$ contains a git commit and the automated tests that are to be run.
	
	The response of possible worlds to $\delta$ always includes the semantics of PHP. Evidence contains the results of the tests and the changes to the git repository (also, maybe a textual representation of the specification). Regarding truth assignments, they map valid formulas of temporal logic to truth values $\true$ and $\false$ and everything else to $\meaningless$; the formulas that are mapped to $\true$ are exactly those that are validated by the program and are consistent with what the possible world allows.
\end{exa}

\begin{exa}[Software evolution \`{a} la~\cite{primiero_theory_2021}] 
	In~\cite{primiero_theory_2021}, a formal system based on belief change is defined to reason about software evolution. In that system, a specification is a propositional theory $\mathcal{S}_m$. The programmer creates an intended implementation $\mathcal{I}$ of the specification. A description $S_i$ of the implementation is then extracted from $\mathcal{I}$, consisting of a set of propositional formulas (not necessarily closed under logical consequence). The theory $\mathcal{S}_m$ and the description $S_i$ can be compared in order to reason whether the implementation satisfies the specification correctly (and whether the entrenchment of properties relative to each other is as desired). In case changes are required, $S_i$ is modified via some specific revision operators (defined in~\cite[Section~3]{primiero_theory_2021}) and this results in a new $S_i^\prime$ which can be abstracted into a new $\mathcal{S}_m^\prime$. The new specification can be then used by the programmer to create a new intended implementation. A simplified form (ignoring entrenchment of properties) of this formal system can be expressed in our framework as follows.
	
	Since $\mathcal{S}_m$ expresses specific behaviours that the intended software system \emph{should} display, its role is similar to that of a specification in our setting. We can define each possible world $w$ to contain a copy of $\mathcal{S}_m$.
	
	$S_i$ represents a model of a specific implementation. In our setting, we can use the conjunction of the formulas in $S_i$ as (part of) the hypothesis at each time instant. Notice that, in such cases, where the hypothesis is meant to contain a description of the program, it is the responsibility of the possible world to check both that $S_i$ is indeed an accurate description of the given program and that it conforms to the specification.
	
	At each step of the process, any (or maybe the first encountered) possible inconsistency between $\mathcal{S}_m$ and $S_i$ is encoded into the evidence returned by the world, so that the programmer can react to it.
	In this context, the revision operators can be considered as \enquote{recipes} that program generators can follow to respond to evidence.
	
	Contrary to~\cite{primiero_theory_2021}, which assumes that $\mathcal{S}_m$ might change at every step of the process, our framework can only allow the update operations to affect $S_i$. While specifications in our framework can change over time, this change is uniform for all different program generators and hence should not depend on the specifics of the workings of any single program generator; more generally, specifications as normative entities (\ie, functional specifications, such as $\mathcal{S}_m$) change quite more rarely and by different kinds of procedures than specifications as descriptive entities (\ie, system specifications, such as $S_i$).
\end{exa}

\noindent The elements obtained during the interaction of program generators and possible worlds can be used for defining other structures, which will be useful when stating properties.

At each time instant $n$ and possible world $w$, the programmer $\delta$ is presented with a truth assignment on possible hypotheses. Hence, there arises a logical language (of program properties)
 \[
\lang{w}{\delta}{n} = \Set{s\in H \suchthat \truth{w}{\delta}{n}(s) \neq \meaningless}.
\] 

Moreover, it is possible to define semantic entailment, both locally (in the given possible world) and globally (over the whole specification).

\begin{defi}
Let $\delta$ be a program generator, $K$ a specification and $n\in\Naturals$.
\begin{enumerate}
\item If $w$ is a possible world, for $r,s\in H$, we define $\implication{r}{s}{\Set{w}}{\delta}{n}$ to be true iff $\truth{w}{\delta}{n}(r) = \true$ implies $\truth{w}{\delta}{n}(s) = \true$.
\item For $r,s \in H$, we define $\implication{r}{s}{K}{\delta}{n}$ to be true iff $\implication{r}{s}{\Set{w}}{\delta}{n}$ for all $w \in K$.
\end{enumerate}
\end{defi}

Hence, every program generator $\delta$ in every possible world $w$ at every $n\in\Naturals$ can be considered to possess a logic having $\lang{w}{\delta}{n}$ as sentences and $K$ as its (intended) models. This motivates the notion of a (logical) \emph{theory} (of program properties): it is formed as the closure (over global semantic entailment) of the current hypothesis, \ie,

\[ \theory{w}{\delta}{K}{n} = \Set{s\in H \suchthat \implication{\hypothesis{w}{\delta}{n}}{s}{K}{\delta}{n}} . \]

Notice that, in general, $\theory{w}{\delta}{K}{n}$ may contain strings marked as false by the current truth assignment; it may even not be a subset of $\lang{w}{\delta}{n}$.

Similarly to the logical definitions, we can define notions that relate to programs; these will be the sets of semantically valid programs. The programs that \enquote{make sense} locally (\ie, in the given possible world) for the program generator are
\[
\programs{w}{\delta}{n} = \Set{\emptystring, \program{w}{\delta}{n}}\cup\Set{p\in P \suchthat \nonProgram{w}{\delta}{n} \neq \semantics{p}{w}{\delta}{n}};
\]
they include the current program, in virtue of it being stated by the program generator, even when it does not actually make sense semantically. On the other hand, the programs that make sense globally (\ie, for the whole specification) for the program generator are
\[
\validPrograms{K}{\delta}{n} = \Set{p \in P \suchthat \forall w \in K,\nonProgram{w}{\delta}{n} \neq \semantics{p}{w}{\delta}{n}}.
\]
Notice that memberhood in these sets is not determined by whether a program meets the specification, but only by whether it is correctly formed in its programming language. The question of whether the program meets its specification is decided by the response of the world to the hypothesis that accompanies the program.

\begin{rem}
$\validPrograms{K}{\delta}{n}  \subseteq \programs{w}{\delta}{n}$ and $\emptystring\notin \validPrograms{K}{\delta}{n}$.
\end{rem}

\subsection{Properties of program generators}

Although any kind of program generator is possible, even for example one that always outputs nonsense, we will usually want to focus our discussion on program generators that actually follow their specification correctly. As in~\cite{arageorgis_relativism_2017}, some properties can be defined that describe good behaviour of a program generator at a specific time instant (and world). Sound methodologies of programming should then correspond to rules that program generators can follow in order to achieve, possibly not from scratch but hopefully from some time onwards, the desired properties.

The simplest such desired property is that the current program must actually be a valid program. Moreover, it is natural to ask that its behaviour is correctly described by the current hypothesis, which, in addition, is compatible with the specification. A more advanced property is that the current theory fully describes (correctly) the behaviour of the current program. These three properties are expressed in the next definition, followed by some of their implications.

Let $\delta$ be a program generator, $K$ a specification, $w \in K$, and $n\in\Naturals$. 

\begin{defi}
The program generator $\delta$:
\begin{enumerate}
\item \emph{is valid at $n$, $w$ given specification $K$} (notation $\isValid{\delta}{w,K}{n}$) iff $\program{w}{\delta}{n} \in \validPrograms{K}{\delta}{n}$,
\item \emph{is correct at $n$, $w$} (notation $\isCorrect{\delta}{w}{n}$) iff $\truth{w}{\delta}{n}(\hypothesis{w}{\delta}{n}) = \true$,
\item \emph{is complete at $n$, $w$ given specification $K$} (notation $\isComplete{\delta}{w,K}{n}$) iff 
\[
\Set{s\in H \suchthat \truth{w}{\delta}{n}(s) = \true} \subseteq \theory{w}{\delta}{K}{n},
\]
\item \emph{is mature at $n$, $w$ given specification $K$} (notation $\isMature{\delta}{w,K}{n}$) iff $\isCorrect{\delta}{w}{n}$, $\isValid{\delta}{w,K}{n}$, and $\isComplete{\delta}{w,K}{n}$.
\end{enumerate}
\end{defi}

\begin{prop}\label{prop:properties-of-theories}
If $\isCorrect{\delta}{w}{n}$, then $\theory{w}{\delta}{K}{n} \subseteq \Set{s\in H\suchthat \truth{w}{\delta}{n}(s)=\true}$ (hence also $\theory{w}{\delta}{K}{n}\subseteq \lang{w}{\delta}{n}$). 
\end{prop}
\begin{proof}
Let $s\in\theory{w}{\delta}{K}{n}$. Then, $\implication{\program{w}{\delta}{n}}{s}{K}{\delta}{n}$ and, a fortiori, $\implication{\hypothesis{w}{\delta}{n}}{s}{\Set{w}}{\delta}{n}$. Since $\truth{w}{\delta}{n}(\hypothesis{w}{\delta}{n}) = \true$, it follows that $\truth{w}{\delta}{n}(s) = \true$.
\end{proof}

\begin{prop}
If $\isMature{\delta}{w,K}{n}$, then
$
 \theory{w}{\delta}{K}{n} = \Set{s\in H\suchthat \truth{w}{\delta}{n}(s)=\true}
$.
\end{prop}
\begin{proof}
Since $\isCorrect{\delta}{w}{n}$, 
Proposition~\ref{prop:properties-of-theories} implies that
$
\theory{w}{\delta}{K}{n}\subseteq \Set{s\in H \suchthat \truth{w}{\delta}{n}(s) = \true}
$.
The fact that $\isComplete{\delta}{w,K}{n}$ completes the argument.
\end{proof}

\subsection{Translations between the outputs of program generators}

Since there might be multiple different program generators working on the same specification, it is important to be able to compare their outputs. Among the items created by the interaction of program generators and worlds, we are mainly interested in comparing the programs and the theories on the programs.
On the other hand, actions and evidence are meant to be esoteric to the interaction between world and program generator, hence those need not concern us during comparisons. These considerations motivate the following definition.

\begin{defi}\label{def:representation-func}
Let $\gamma$ and $\delta$ be two program generators with common specification $K$. For every possible worlds $w_1, w_2\in K$ and stages $m$, $n$ of development, we will say that the output of $\gamma$ is \emph{translatable} to the output of $\delta$ (notation $\representable{\gamma}{\delta}{m}{n}{w_1}{w_2}$) iff there exists a \emph{translation function} $\function{\representation{w_1}{w_2}{\gamma}{\delta}{m}{n}}{\programs{w_1}{\gamma}{m}\cup \lang{w_1}{\gamma}{m}}{P\cup H}$ such that
\begin{enumerate}
\item\label{def:representation-func:preserves-bottom} $\representation{w_1}{w_2}{\gamma}{\delta}{m}{n}(\emptystring) = \emptystring$,
\item\label{def:representation-func:maps-lang} $\representation{w_1}{w_2}{\gamma}{\delta}{m}{n}(\lang{w_1}{\gamma}{m}) \subseteq \lang{w_2}{\delta}{n}$
\item\label{def:representation-func:maps-programs} $\representation{w_1}{w_2}{\gamma}{\delta}{m}{n}(\programs{w_1}{\gamma}{m}) \subseteq \programs{w_2}{\delta}{n}$,
\item\label{def:representation-func:maps-program}  $\representation{w_1}{w_2}{\gamma}{\delta}{m}{n}(\program{w_1}{\gamma}{m}) \simulation{w_2}{\delta}{n} \program{w_2}{\delta}{n}$,
\item\label{def:representation-func:preserves-truth} for $r, s\in \lang{w_1}{\gamma}{m}$, 
if $\implication{r}{s}{\Set{w_1}}{\gamma}{m}$, then $\implication{\representation{w_1}{w_2}{\gamma}{\delta}{m}{n}(r)}{\representation{w_1}{w_2}{\gamma}{\delta}{m}{n}(s)}{\Set{w_2}}{\delta}{n}$,
\item\label{def:representation-func:preserves-semantics} for $p_1, p_2\in\programs{w_1}{\delta}{m}$, if $p_1 \simulation{w_1}{\gamma}{m}p_2$ then $\representation{w_1}{w_2}{\gamma}{\delta}{m}{n}(p_1) \simulation{w_2}{\delta}{n} \representation{w_1}{w_2}{\gamma}{\delta}{m}{n}(p_2)$,
\item\label{def:representation-func:maps-theory} $\representation{w_1}{w_2}{\gamma}{\delta}{m}{n}(\theory{w_1}{\gamma}{K}{m}) \subseteq \theory{w_2}{\delta}{K}{n}$,
\item\label{def:representation-func:maps-valid-programs} $\representation{w_1}{w_2}{\gamma}{\delta}{m}{n}(\validPrograms{K}{\gamma}{m}) \subseteq \validPrograms{K}{\delta}{n}$.
\end{enumerate}
\end{defi}

\noindent Points~\ref{def:representation-func:maps-lang} and~\ref{def:representation-func:maps-programs} specify that the translation maps sentences to sentences and programs to programs, while points~\ref{def:representation-func:maps-theory} and~\ref{def:representation-func:maps-valid-programs} require that the same happens for sentences in theories and for valid programs. Since there is a designated string, the empty one, point~\ref{def:representation-func:preserves-bottom} takes care of it.

Point~\ref{def:representation-func:preserves-semantics} requires that a translation preserve program semantics. Of course, source code meaningful in one setting may not be meaningful in another (for example, $\gamma$ may operate in one programming language and $\delta$ in another). Most importantly, however, programs may be mapped to different kinds of mathematical objects in the two settings (for example, $\gamma$ may operate in a functional paradigm and $\delta$ in a procedural paradigm), hence we can only express preservation of program semantics indirectly, by the (weaker) condition that translation must preserve the simulation relation. 

Point~\ref{def:representation-func:preserves-truth} requires that a translation preserve logical entailment. Since, contrary to the situation with programs, truth assignments have the same range in all settings, we could have opted for the stronger condition $\truth{w_1}{\gamma}{m}(s) = \truth{w_2}{\delta}{n}(\representation{w_1}{w_2}{\gamma}{\delta}{m}{n}(s))$; however this would have two important ramifications: firstly, it would complicate adding more truth values (maybe for only some of the program generators), and, secondly, it would prevent more mature agents from correctly reinterpreting the false theories of their less mature counterparts.

Finally, point~\ref{def:representation-func:maps-programs} requires  that, for a successful translation, the program of $\delta$ must be able to perform (at least) the operations of the translation of the program of $\gamma$ (again, since the simulation relation might be different in the two worlds, we can only compare indirectly).

\begin{rem}
Translation functions compose: if $\representation{w_1}{w_2}{\beta}{\gamma}{l}{m}$ and $\representation{w_2}{w_3}{\gamma}{\delta}{m}{n}$ are translation functions, then $ \representation{w_1}{w_2}{\beta}{\gamma}{l}{m} \catcompose \representation{w_2}{w_3}{\gamma}{\delta}{m}{n}$ is a translation function.
Since the identity function is a translation function, it follows that, given a specification $K$, a category of translation functions arises.
\end{rem}

\begin{cor}
If $\representation{w}{w^\prime}{\gamma}{\delta}{m}{n}$ is a translation function and $\representation{w^\prime}{w}{\delta}{\gamma}{n}{m}$ is another translation function, then
\begin{enumerate}
\item $(\representation{w}{w^\prime}{\gamma}{\delta}{m}{n} \catcompose \representation{w^\prime}{w}{\delta}{\gamma}{n}{m})(\program{w}{\gamma}{m})  \simulation{w}{\gamma}{m} \program{w}{\gamma}{m}$,
\item $(\representation{w}{w^\prime}{\gamma}{\delta}{m}{n} \catcompose \representation{w^\prime}{w}{\delta}{\gamma}{n}{m})(\theory{w}{\gamma}{K}{m}) \subseteq  \theory{w}{\gamma}{K}{m}$.
\end{enumerate}
\end{cor}

\begin{thm}
Let $\gamma$ and $\delta$ be two program generators with common specification $K$. For every possible world $w\in K$ and stages $m$, $n$ of development, if $\isMature{\gamma}{w_1,K}{m}$ and $\isMature{\delta}{w_2,K}{n}$, then $\representable{\gamma}{\delta}{m}{n}{w_1}{w_2}$ or $\representable{\delta}{\gamma}{n}{m}{w_2}{w_1}$.
\end{thm}
\begin{proof}
We will construct a (trivial) translation function $f$.

Of course, by maturity, both $\lang{w_1}{\gamma}{m}$ and $\lang{w_2}{\delta}{n}$ contain some true statement, namely, $\hypothesis{w_1}{\gamma}{m}$ and $\hypothesis{w_2}{\delta}{n}$ respectively. Let $s_{\true}^{\gamma}$ be some true statement in $\lang{w_1}{\gamma}{m}$ (any such statement is enough for this proof, but if the set of all true statements is finitely generated by a set $S$, we could set $s_{\true}^{\gamma} = \bigwedge S$; similar considerations apply for $s_{\true}^{\delta}$) and $s_{\true}^{\delta}$ some true statement in $\lang{w_2}{\delta}{n}$. If there is no false statement in $\lang{w_1}{\gamma}{m}$ and no false statement in $\lang{w_2}{\delta}{n}$, then we can define $\forall s\in \lang{w_1}{\gamma}{m}, f(s) = s_{\true}^{\delta}$. Otherwise, at least one of the languages contains a false statement; without loss of generality, let that language be $\lang{w_2}{\delta}{n}$ and let $s_{\false}^{\delta}\in\lang{w_2}{\delta}{n}$ be the false statement. Then, we can define
\[
\forall s \in \lang{w_1}{\gamma}{m}, f(s) = \begin{cases}
s_{\true}^{\delta} \casesif \truth{w_1}{\gamma}{m}(s) = \true\\
s_{\false}^{\delta} \casesowise
\end{cases}
\]
In both cases, $\truth{w_1}{\gamma}{m}(s) = \truth{w_2}{\delta}{n}(f(s))$ for all $s \in \lang{w_1}{\gamma}{m}$ and hence $f$ trivially validates point~\ref{def:representation-func}.\ref{def:representation-func:preserves-truth}. By definition, \ref{def:representation-func}.\ref{def:representation-func:maps-lang} is also valid for $f$.
Finally, by correctness of $\gamma$, $\theory{w_1}{\gamma}{K}{m} \subseteq \Set{s\in S \suchthat \truth{w_1}{\gamma}{m}(s) = \true}$ and hence $f(\theory{w_1}{\gamma}{K}{m}) = \Set{s_{\true}^{\delta}}$, so we obtain~\ref{def:representation-func}.\ref{def:representation-func:maps-theory} by completeness of $\delta$. 

As for programs, define
\[
\forall p \in \programs{w_1}{\gamma}{m}, f(p) = \begin{cases}
\program{w_2}{\delta}{n} \casesif \exists p^\prime \in \validPrograms{K}{\gamma}{m}: p^\prime \simulation{w_1}{\gamma}{m} p\\
\emptystring \casesowise
\end{cases}
\]
By definition, \ref{def:representation-func}.\ref{def:representation-func:preserves-bottom} and~\ref{def:representation-func}.\ref{def:representation-func:maps-programs} hold for $f$. Since $\gamma$ is mature and the simulation relations are reflexive, we obtain~\ref{def:representation-func}.\ref{def:representation-func:maps-program}. By maturity of $\delta$  and the reflexivity of the simulation relation (of $\gamma$), \ref{def:representation-func}.\ref{def:representation-func:maps-valid-programs} is also true. For~\ref{def:representation-func}.\ref{def:representation-func:preserves-semantics}, let $p_1, p_2 \in \programs{w_1}{\gamma}{m}$ with $p_1 \simulation{w_1}{\gamma}{m} p_2$; if $f(p_1) = \emptystring$, then we are done; else, $p^\prime \simulation{w_1}{\gamma}{m} p_1$ for some $p^\prime \in \validPrograms{K}{\gamma}{m}$, hence, by transitivity of the simulation relation (of $\gamma$), $f(p_1) = f(p_2)$ and reflexivity of the simulation relation (of $\delta$) completes the argument.
\end{proof}

\noindent The triviality of the function constructed in the preceding theorem lies on the fact that it has only four output values, thus eliminating, for example, any differences in program behaviour and any relations between properties. However, taking care of that would require possessing more information regarding the inner structure of programs, theories, and the workings of formulating theories of programs, hence it can only happen in more specific instantiations of the framework.

By taking the conditions of the theorem in the limit, we obtain that two program generators that asymptotically arrive at mature programs with respect to a shared specification will asymptotically produce programs and theories one of which can be translated into the other. 

\begin{cor}
Let $\gamma$ and $\delta$ be two program generators with common specification $K$. If for every $w \in K$, $\achievesMaturity{\gamma}{w,K}{n}$ and $\achievesMaturity{\delta}{w,K}{n}$, then for every $w \in K$, 
\[
\trueInLimit{\representable{\gamma}{\delta}{n}{n}{w}{w} \orelse \representable{\delta}{\gamma}{n}{n}{w}{w}}{n}.
\]
\end{cor}

\section{Conclusion}\label{section:conclusion}

Similar to~\cite{arageorgis_relativism_2017}, the central elements of the framework on its technical level (corresponding to the central elements of specifications and software developers on the conceptual level) are \emph{possible worlds} and \emph{program generators} and we are studying their interaction. At each time instance, a program generator processes the evidence up to the given time and produces a program, a hypothesis, and an action. The possible world responds with a truth assignment, program semantics, and some new evidence; the new evidence can be appended to the existing and be fed into the program generator at the next time instance, continuing the interaction.

Using these elements above, we have defined an abstract notion of translatability between the outputs of programmers working in different settings. As a token of the importance of such a notion, we have shown, similarly to~\cite{arageorgis_relativism_2017}, that two program generators starting with common specification (\ie, set of possible worlds) and writing correct programs relative to that specification will end up in programs one of which can be (trivially) translated into the other. Specific instantiations of the framework can refine the definition of translatability we have provided and thus, by adapting our proof accordingly, arrive at translations with more specific structure.

The fact that we could adapt a framework targetted at describing natural science and end up with a framework that describes software development and still highly resembles its origin hints on the similarities between the two endeavours. However, notice that, in addition to introducing new elements to account for the technical artefacts that are produced by programmers, we had to reinterpret some terms of the original framework in order to adapt it for our purposes, most notably \emph{background knowledge} was reinterpreted as \emph{specification}; this is a hint of (one of) the key differences between the two endeavours.

Of course, as already stated in Section~\ref{section:philosophy-of-cs}, our framework does not handle all the aspects of software development, let alone of the development of full computational systems. We have commented on a few of the ways that it might be extended to accommodate more such aspects. In addition, the reasoning on logics of program generators might benefit if it is described via the theory of institutions, an abstraction of model theory based on category theory; \cite{dimarogkona-thesis} has already attempted such a modification of the original framework of~\cite{arageorgis_relativism_2017}. Moreover, other kinds of mappings between the outputs of program generators could be considered instead of translations, such as conceptual blending~\cite{goguen_introduction_1999}; this might model the creative process of merging interesting ideas from one program to the other while the programs are still expressed in different formalisms.

\bibliographystyle{alphaurl}
\bibliography{paper.bib}

\begin{thebibliography}{DMP18}

\bibitem[AB89]{arageorgis_demarcating_1989}
Aristidis Arageorgis and Aristides Baltas.
\newblock Demarcating technology from science: {Problems} and problem solving
  in technology.
\newblock {\em Journal for General Philosophy of Science}, 20(2):212--229,
  September 1989.
\newblock \href {https://doi.org/10.1007/BF01801475}
  {\path{doi:10.1007/BF01801475}}.

\bibitem[Ang13]{angius_abstraction_2013}
Nicola Angius.
\newblock Abstraction and {Idealization} in the {Formal} {Verification} of
  {Software} {Systems}.
\newblock {\em Minds and Machines}, 23(2):211--226, May 2013.
\newblock \href {https://doi.org/10.1007/s11023-012-9289-8}
  {\path{doi:10.1007/s11023-012-9289-8}}.

\bibitem[Ang14]{angius_problem_2014}
Nicola Angius.
\newblock The {Problem} of {Justification} of {Empirical} {Hypotheses} in
  {Software} {Testing}.
\newblock {\em Philosophy \& Technology}, 27(3):423--439, September 2014.
\newblock \href {https://doi.org/10.1007/s13347-014-0159-6}
  {\path{doi:10.1007/s13347-014-0159-6}}.

\bibitem[Ang20]{angius_mutual_2020}
Nicola Angius.
\newblock On the {Mutual} {Dependence} {Between} {Formal} {Methods} and
  {Empirical} {Testing} in {Program} {Verification}.
\newblock {\em Philosophy \& Technology}, 33(2):349--355, June 2020.
\newblock \href {https://doi.org/10.1007/s13347-019-00364-9}
  {\path{doi:10.1007/s13347-019-00364-9}}.

\bibitem[AP18]{angius_logic_2018}
Nicola Angius and Giuseppe Primiero.
\newblock The logic of identity and copy for computational artefacts.
\newblock {\em Journal of Logic and Computation}, 28(6):1293--1322, September
  2018.
\newblock \href {https://doi.org/10.1093/logcom/exy012}
  {\path{doi:10.1093/logcom/exy012}}.

\bibitem[AP22]{angius_copying_2022}
Nicola Angius and Giuseppe Primiero.
\newblock Copying safety and liveness properties of computational artefacts.
\newblock {\em Journal of Logic and Computation}, August 2022.
\newblock \href {https://doi.org/10.1093/logcom/exac053}
  {\path{doi:10.1093/logcom/exac053}}.

\bibitem[APT21]{plato_philosophy_2021}
Nicola Angius, Giuseppe Primiero, and Raymond Turner.
\newblock The {Philosophy} of {Computer} {Science}.
\newblock In Edward~N. Zalta, editor, {\em The {Stanford} {Encyclopedia} of
  {Philosophy}}. Metaphysics Research Lab, Stanford University, spring 2021
  edition, 2021.
\newblock URL:
  \url{https://plato.stanford.edu/archives/spr2021/entries/computer-science/}.

\bibitem[Ara17]{arageorgis_relativism_2017}
Aristidis Arageorgis.
\newblock Relativism, translation, and the metaphysics of realism.
\newblock {\em Philosophical Studies}, 174(3):659--680, March 2017.
\newblock \href {https://doi.org/10.1007/s11098-016-0702-7}
  {\path{doi:10.1007/s11098-016-0702-7}}.

\bibitem[CS14]{git}
Scott Chacon and Ben Straub.
\newblock {\em Pro git}.
\newblock Apress, 2014.

\bibitem[Dim]{dimarogkona-thesis}
Maria~M. Dimarogkona.
\newblock {PhD} thesis, School of Applied Mathematical and Physical Science,
  National Technical University of Athens.
\newblock To appear.

\bibitem[DMP18]{de_mol_what_2018}
Liesbeth De~Mol and Giuseppe Primiero.
\newblock What is a program? {Historical} and {Philosophical} perspectives.
\newblock {\em The Reasoner}, 12(2):11--12, February 2018.

\bibitem[Fit06]{fitting_bilattices_2006}
Melvin Fitting.
\newblock Bilattices {Are} {Nice} {Things}.
\newblock In T.~Bolander, V.~Hendricks, and S.~A. Pedersen, editors, {\em
  Self-{Reference}}. CSLI Publications, 2006.

\bibitem[Flo16]{floridi_method_2016}
Luciano Floridi.
\newblock The method of abstraction.
\newblock In Luciano Floridi, editor, {\em The {Routledge} handbook of
  philosophy of information}, Routledge handbooks in philosophy, pages 50--56.
  Routledge, Taylor \& Francis Group, London New York, 2016.

\bibitem[Gog99]{goguen_introduction_1999}
Joseph Goguen.
\newblock An {Introduction} to {Algebraic} {Semiotics}, with {Application} to
  {User} {Interface} {Design}.
\newblock In Chrystopher~L. Nehaniv, editor, {\em Computation for {Metaphors},
  {Analogy}, and {Agents}}, Lecture {Notes} in {Computer} {Science}, pages
  242--291, Berlin, Heidelberg, 1999. Springer.
\newblock \href {https://doi.org/10.1007/3-540-48834-0_15}
  {\path{doi:10.1007/3-540-48834-0_15}}.

\bibitem[Gru11]{gruner_problems_2011}
Stefan Gruner.
\newblock Problems for a {Philosophy} of {Software} {Engineering}.
\newblock {\em Minds and Machines}, 21(2):275--299, May 2011.
\newblock \href {https://doi.org/10.1007/s11023-011-9234-2}
  {\path{doi:10.1007/s11023-011-9234-2}}.

\bibitem[KJG92]{kelly_reliability_1992}
Kevin~T. Kelly, Cory Juhl, and Clark Glymour.
\newblock Reliability, {Realism}, and {Relativism}.
\newblock 1992.

\bibitem[PRC21]{primiero_theory_2021}
Giuseppe Primiero, Franco Raimondi, and Taolue Chen.
\newblock A theory of change for prioritised resilient and evolvable software
  systems.
\newblock {\em Synthese}, 198(23):5719--5744, October 2021.
\newblock \href {https://doi.org/10.1007/s11229-019-02305-7}
  {\path{doi:10.1007/s11229-019-02305-7}}.

\bibitem[Pri16]{primiero_information_2016}
Giuseppe Primiero.
\newblock Information in the {Philosophy} of {Computer} {Science}.
\newblock In Luciano Floridi, editor, {\em The {Routledge} handbook of
  philosophy of information}, Routledge handbooks in philosophy, pages 90--106.
  Routledge, Taylor \& Francis Group, London New York, June 2016.

\bibitem[Tur11]{turner_specification_2011}
Raymond Turner.
\newblock Specification.
\newblock {\em Minds and Machines}, 21(2):135--152, May 2011.
\newblock \href {https://doi.org/10.1007/s11023-011-9239-x}
  {\path{doi:10.1007/s11023-011-9239-x}}.

\bibitem[Tur18]{turner_computational_2018}
Raymond Turner.
\newblock {\em Computational {Artifacts}: {Towards} a {Philosophy} of
  {Computer} {Science}}.
\newblock Springer, Berlin, Heidelberg, 1st edition, 2018.
\newblock \href {https://doi.org/10.1007/978-3-662-55565-1}
  {\path{doi:10.1007/978-3-662-55565-1}}.

\bibitem[Tur20]{turner_computational_2020}
Raymond Turner.
\newblock Computational {Intention}.
\newblock {\em Studies in Logic, Grammar and Rhetoric}, 63(1):19--30, September
  2020.
\newblock \href {https://doi.org/10.2478/slgr-2020-0025}
  {\path{doi:10.2478/slgr-2020-0025}}.

\bibitem[ZJ97]{zave_four_1997}
Pamela Zave and Michael Jackson.
\newblock Four dark corners of requirements engineering.
\newblock {\em ACM Transactions on Software Engineering and Methodology},
  6(1):1--30, January 1997.
\newblock \href {https://doi.org/10.1145/237432.237434}
  {\path{doi:10.1145/237432.237434}}.

\end{thebibliography}
\end{document}